\documentclass[a4paper,10pt,reqno]{amsart}

\usepackage{etex}
\usepackage{amssymb}
\usepackage{amsmath}
\usepackage{amsfonts}
\usepackage{amsthm}
\usepackage{mathrsfs}
\usepackage{eucal}
\usepackage[safe]{tipa}
\usepackage[normalem]{ulem}
\usepackage{cases}
\usepackage{xr}
\usepackage{color}
\usepackage{enumitem}

\usepackage{ulem}

\usepackage{slashed}
\usepackage{textalpha}

\usepackage{lpic}

\makeatletter
\newcommand\mathcircled[1]{%
  \mathpalette\@mathcircled{#1}%
}
\newcommand\@mathcircled[2]{%
  \tikz[baseline=(math.base)] \node[draw,circle,inner sep=1pt] (math) {$\m@th#1#2$};%
}
\makeatother

\usepackage{url}
\usepackage[linktocpage=true]{hyperref}

\hypersetup{
    colorlinks=true,
    linkcolor=blue,
    filecolor=blue,      
    urlcolor=blue,
}

\newcommand{\Scal}S

\theoremstyle{plain}
\newtheorem{thm}{Theorem}[section]

\newtheorem{example}[thm]{Example}

\newtheorem{prop}[thm]{Proposition}
\newtheorem{remark}[thm]{Remark}

\def\d{{\rm d}}

\title[Charges, complex structures, and perturbations]{Charges, complex structures, and perturbations of instantons}

\author[L. Andersson]{Lars Andersson}
\address[Lars Andersson]{University of Potsdam, 14476 Potsdam, Germany}
\email{lars.andersson@uni-potsdam.de}

\author[B. Araneda]{Bernardo Araneda}
\address[Bernardo Araneda]{School of Mathematics and Maxwell Institute for Mathematical Sciences, University of Edinburgh, EH9 3FD, United Kingdom,  \newline
Max-Planck-Institut f\"ur Gravitationsphysik (Albert-Einstein-Institut), 
Am M\"uhlenberg 1, D-14476 Potsdam, Germany}
\email{baraneda@ed.ac.uk, bernardo.araneda@aei.mpg.de}

\allowdisplaybreaks[2]
\numberwithin{equation}{section}

\begin{document}

\begin{abstract}
Hermitian non-K\"ahler Einstein 4-manifolds have a quasi-locally conserved charge associated to spin-lowering via Killing spinors, and corresponding to a parameter of the moduli space. This charge is evaluated for all explicitly known examples of gravitational instantons. 
Generic gravitational perturbations are shown to admit a closed 2-form that measures the perturbation to this charge, generalizing previous Lorentzian results on the linearized mass of perturbed Kerr black holes.
\end{abstract} 

\date{\today}
\maketitle

\section{Introduction}

In Euclidean quantum gravity, the dominant contribution to the gravitational path integral is expected to come from metrics near the minima of the Euclidean action \cite{GibbonsHawkingPerry}, that is from gravitational instantons and small deformations of them. In the one-loop expansion, deformations of a classical solution $(M,g_{ab})$ represent quantum corrections \cite{GibbonsPerry}, and in the case of perturbations satisfying the linearized Einstein equations, they provide information about the moduli space $\mathscr{E}(M)$ of  instantons, describing the structure of the tangent space of $\mathscr{E}(M)$ at $g_{ab}$ and allowing to analyze questions such as integrability, infinitesimal rigidity \cite{Andersson:2024wtn, AA} and spectral stability (e.g. \cite{dahlkroncke, BiquardOzuch}); cf. \cite{Besse} for the compact case. 

The study of gravitational instantons is further motivated by the recently discovered Chen-Teo instanton \cite{ChenTeo1}, a Ricci-flat, asymptotically flat (AF) solution that gives a counter-example to the Euclidean black hole uniqueness conjecture \cite{Gibbons1980, Gibbons:1979xm} and 
contributes, in addition to flat space and to the Schwarzschild and Kerr instantons, to the possible semi-classical states of the gravitational field with AF boundary conditions \cite{Lapedes:1980st}. 
Recent work by Li and Sun \cite{LiSun} shows the existence of a new infinite family of AF gravitational instantons,
and the classification problem remains largely open.
Previous classification results include \cite{Biquard:2021gwj, Aksteiner:2023djq, Li:2023jlq}.

For the vacuum Einstein equations with or without cosmological constant, 
boundary conditions of interest in quantum gravity include, in addition to AF, asymptotcally locally flat (ALF),
asymptotically locally Euclidean (ALE), and compact manifolds. All explicitly known instantons are (at least locally\footnote{Our definition of instantons here allows a non-vanishing cosmological constant. This includes the 4-sphere $S^4$, which does not admit a global complex structure.}) Hermitian \cite{Aksteiner:2021fae}. 
Furthermore, they are conformally K\"ahler \cite{Derdzinski, Lebrun95}. This means that there is a metric $\hat{g}_{ab}=\Omega^2 g_{ab}$ 
and a self-dual 2-form $\hat\omega_{ab}$ which is parallel under the Levi-Civita connection $\hat\nabla_{a}$ of $\hat{g}_{ab}$,
\begin{align}\label{Kahler}
 \hat\nabla_{a}\hat\omega_{bc} = 0.
\end{align}

In this paper, we shall see that for Hermitian instantons, if the part of the Weyl tensor that has the same orientation as the complex structure ($\Psi_{ABCD}$ in 2-spinor notation, cf. section \ref{sec:preliminaries} for conventions) is non-vanishing, then one of the moduli corresponds to a special ``charge'' associated to the gravitational field. This is because the K\"ahler form is actually the spin-lowered Weyl spinor
\begin{align}\label{KahlerForm0}
 \hat\omega_{ab} = \Psi_{ABCD}K^{CD}\epsilon_{A'B'},
\end{align}
where $K^{AB}$ is a real solution to the 2-index twistor equation, see section \ref{sec:IdentityForCharges} below. The gravitational charges referred to above are then the periods of the closed 2-form \eqref{KahlerForm0} over 2-cycles $S$ in $M$,
\begin{align}\label{charge}
 Q[S] = \int_{S}\hat\omega.
\end{align}
Since $\d\hat\omega=0$, \eqref{charge} is invariant under changes of representative of the homology class $[S]$, so the charge is both geometrical (it depends on the curvature of $g_{ab}$) and topological (it depends on the homology groups of $M$).

The above result is the Riemannian analog of a Lorentzian conservation law first found by Jordan, Ehlers, and Sachs \cite{JordanEhlersSachs}, and then reformulated in \cite{Aksteiner:2013rq} in terms of spin-lowering \cite[§6.4]{PR2} for the Kerr black hole, where the charge \eqref{charge} corresponds to the mass parameter and the conservation law is intuitively clear: 2-spheres surrounding the black hole are non-contractible, and the mass contained within any two of them is the same.
However, for gravitational instantons, the evaluation of \eqref{charge} and its interpretation are considerably more subtle, since the topology may be substantially different from the Kerr topology $\mathbb{R}^2\times S^2$, and it is not obvious how to find the 2-cycles (i.e. the 2-spheres in Kerr).
As we shall see, examples of \eqref{charge} include: the NUT-charge/mass of the ALF Taub-bolt instanton on $\mathbb{CP}^2\setminus \{\rm pt.\}$, the only free parameter of the ALE Eguchi-Hanson instanton on $T^{*}\mathbb{CP}^1$, the cosmological constant for the compact instantons $\mathbb{CP}^2$ and $\mathbb{CP}^2\#\overline{\mathbb{CP}}{}^2$, and the two parameters of the AF Chen-Teo instanton on $\mathbb{CP}^2\setminus S^1$. The fact that the two parameters of Chen-Teo are integrals of the K\"ahler form is to be contrasted with the Kerr instanton, where \eqref{charge} gives only the mass.

We shall prove that generic gravitational perturbations of Hermitian instantons, with or without cosmological constant $\lambda$ and with any asymptotic structure, admit a 2-form $\delta\hat\omega_{ab}$ (given by \eqref{pertclosedform2}) which is closed, $\d\:\delta\hat\omega=0$, and whose integral over 2-cycles measures the perturbation to \eqref{charge},
\begin{align}\label{perturbedcharge}
 \delta{Q}[S] = \int_{S}\delta\hat\omega.
\end{align}
Moreover, the result \eqref{perturbedcharge} is gauge invariant, i.e. invariant under infinitesimal diffeomorphisms; see eq. \eqref{gaugetr} below. 
The new conservation law \eqref{perturbedcharge} for gravitational perturbations of instantons also has a Lorentzian counterpart, which was first found by Fackerell in \cite{Fackerell} (cf. also \cite{Fayos}),
and then formulated in terms of spin-lowering in \cite{Aksteiner:2013rq} and used to show that the analog of \eqref{perturbedcharge} gives the linearized mass of a perturbed Kerr black hole \cite[Theorem 5.1]{Aksteiner:2013rq}.

The result \eqref{perturbedcharge} is valid for any asymptotic structure, and in this paper we shall not restrict ourselves to any boundary conditions. However, in the ALF case, it was conjectured \cite{Andersson:2024wtn} that the only possible gravitational perturbations are moduli variations. 
In a subsequent paper \cite{AA} we prove that this is indeed the case. 
Physically, this means that there are no higher-order quantum corrections to the instantons, i.e. the one-loop expansion remains valid.

\subsection*{Overview}
In section \ref{sec:preliminaries} we introduce some background material and set our notation and conventions.
In section \ref{sec:conservationlaw} we prove an identity for generic Riemannian 4-manifolds (Prop. \ref{prop:identityforms}), and use it to derive the closed 2-form $\delta\hat\omega_{ab}$ and the linearized charge \eqref{perturbedcharge}; see Theorem \ref{thm:closed2form}.
Section \ref{sec:Examples} contains many examples of the charge \eqref{charge} for different gravitational instantons (AF, ALF, ALE, compact).

\section{Preliminaries and notation}
\label{sec:preliminaries}

Throughout the article, $(M,g_{ab})$ denotes a four-dimensional, smooth, orientable Riemannian manifold, with metric signature $(++++)$, Levi-Civita connection $\nabla_{a}$, volume form $\epsilon_{abcd}$, and Hodge star operator $*$. We shall use, {\em locally}, the 2-spinor formalism \cite{PR1,PR2}, adapted to Riemannian signature. We do not assume a global spin structure. 

The (locally defined) unprimed and primed spinor bundles are $\mathbb{S}$ and $\mathbb{S}'$. The symplectic forms $\epsilon_{AB}$ and $\epsilon_{A'B'}$ and their inverses are used to raise and lower spinor indices. 
An unprimed spin dyad is $(o^A,\iota^A)$ with $\epsilon_{AB}o^{A}\iota^{B}=1$ and $\iota^{A}=o^{\dagger A}$, and a primed spin dyad is $(\alpha^{A'},\beta^{A'})$ with $\epsilon_{A'B'}\alpha^{A'}\beta^{B'}=1$ and $\beta^{A'}=\alpha^{\dagger A'}$. Defining the null tetrad $\ell^{a}=o^{A}\alpha^{A'}$, $n^{a}=\iota^{A}\beta^{A'}$, $m^{a}=o^{A}\beta^{A'}$, $\tilde{m}^{a}=\iota^{A}\alpha^{A'}$ (see e.g. \cite[(3.1.14)]{PR1}), we have $n^{a}=\bar{\ell}^{a}$ and $\tilde{m}^{a}=-\bar{m}^{a}$, and the non-trivial inner products are $\ell^a n_a=1=-m^a\tilde{m}_a$. The real and imaginary parts of the null tetrad define four real vectors $\{e^{a}_{\bf a}\}$, ${\bf a}=0,...,3$, as follows: $e^{a}_{0}=\frac{1}{\sqrt{2}}(\ell^{a}+n^{a})$, 
$e^{a}_{1}=i\frac{1}{\sqrt{2}}(\ell^{a} - n^{a})$, 
$e^{a}_{2}=\frac{1}{\sqrt{2}}(m^{a}-\tilde{m}^{a})$, 
$e^{a}_{3}=i\frac{1}{\sqrt{2}}(m^{a}+\tilde{m}^{a})$. 
(This is the analog of \cite[Eq. (3.1.20)]{PR1}.)
The dual basis, denoted $e_{a}^{\bf a}$, is defined by $e_{a}^{\bf a}e^{a}_{\bf b}=\delta^{\bf a}_{\bf b}$. A computation gives:
\begin{align*}
 \epsilon_{AB}\epsilon_{A'B'} = 2(\ell_{(a}n_{b)}-m_{(a}\tilde{m}_{b)}) 
 = +\delta_{\bf ab}e_{a}^{\bf a}e_{b}^{\bf b} = +g_{ab}.
\end{align*}
Thus the reality conditions on the null tetrad imposed above, which are in turn defined by the Riemannian spinor conjugation $\dagger$, give the correct signature $(++++)$ for the metric. (Note that an analogous calculation but with Lorentzian reality conditions gives signature $(+---)$, cf. footnote in \cite[p. 121]{PR1}.)

The space of 2-forms in $M$, $\Lambda^2$, splits under $*$ as $\Lambda^2=\Lambda^{2}_{+}\oplus\Lambda^{2}_{-}$, where $\Lambda^{2}_{\pm}$ corresponds to the eigenvalue $\pm 1$ of $*$. Elements of $\Lambda^{2}_{+}$ ($\Lambda^{2}_{-}$) are called self-dual (anti-self-dual) 2-forms. Our convention for the volume form is $\epsilon_{0123}=+1$, i.e. $\epsilon_{abcd}=4!e^{0}_{[a}e^{1}_{b}e^{2}_{c}e^{3}_{d]}$, or equivalently $\pmb{\epsilon}=+e^0 \wedge e^1 \wedge e^2 \wedge e^3$. With this convention, one can check that a basis of real self-dual 2-forms is:
\begin{subequations}\label{SD2forms}
\begin{align}
 & Z^{1}_{ab} = \tfrac{1}{\sqrt{2}}(e^0\wedge e^3 + e^1\wedge e^2)_{ab} = \tfrac{i}{\sqrt{2}}(o_Ao_B - \iota_A\iota_B)\epsilon_{A'B'}, \\
 & Z^{2}_{ab} = \tfrac{1}{\sqrt{2}}(e^0\wedge e^2 - e^1\wedge e^3)_{ab} = \tfrac{1}{\sqrt{2}}(o_Ao_B + \iota_A\iota_B)\epsilon_{A'B'}, \\
 & Z^{3}_{ab} = \tfrac{1}{\sqrt{2}}(e^0\wedge e^1 + e^2\wedge e^3)_{ab} = -\tfrac{i}{\sqrt{2}} (o_{A}\iota_{B}+\iota_{A}o_{B})\epsilon_{A'B'}. \label{Z3}
\end{align}
\end{subequations}
Self-dual 2-forms are, thus, of the form $\omega^{+}_{ab}=\psi_{AB}\epsilon_{A'B'}$, while anti-self-dual 2-forms are $\omega^{-}_{ab}=\varphi_{AB}\epsilon_{A'B'}$. In other words, $\Lambda^{2}_{+}\cong \mathbb{S}\odot\mathbb{S}$ and $\Lambda^{2}_{-}\cong \mathbb{S}'\odot\mathbb{S}'$. Note that {\em this convention is the opposite of \cite{PR1, PR2}}.

The 2-forms \eqref{SD2forms} give an orthonormal basis of $\Lambda^{2}_{+}$. In terms of spinors, this can be phrased in terms of the real `metric' on $\mathbb{S}^2=\mathbb{S}\odot\mathbb{S}$ given by
\begin{align}\label{metricforms}
I_{ABCD}=\epsilon_{A(C}\epsilon_{|B|D)}, 
\end{align}
cf. \cite[Eq. (8.3.2)]{PR2}, with inverse $I^{ABCD}$. 
Writing \eqref{SD2forms} as $Z^{i}_{ab}=\zeta^{i}_{AB}\epsilon_{A'B'}$ ($i=1,2,3$), the spinors $\zeta^{i}_{AB}$ are real, symmetric, and satisfy
\begin{align}
 I^{ABCD}\zeta^{i}_{AB}\zeta^{j}_{CD} = \delta^{ij}, \qquad i=1,2,3.
\end{align}
It also follows that $\delta_{ij}\zeta^{i}_{AB}\zeta^{j}_{CD}=I_{ABCD}$. Thus $\{\zeta^{i}_{AB}\}$ gives a real, orthonormal basis of $\mathbb{S}^2$.

We also note that the self-dual 2-forms $\omega^{i}_{ab}=\sqrt{2}Z^{i}_{ab}$ give (locally) an almost-hyper-Hermitian structure, i.e. they satisfy the quaternion algebra. In particular, each of $\omega^{i}_{ab}$ defines an almost-complex structure $J^{i}$ compatible with the metric, via $\omega^{i}=g(J^{i}\cdot,\cdot)$.

The Levi-Civita connection of $g_{ab}$ gives a connection in $\mathbb{S}^2$, also denoted $\nabla_{a}$. Since $\nabla_{a}:\Gamma(\mathbb{S}^2)\to\Gamma(T^{*}M\otimes\mathbb{S}^2)$, the covariant derivative of $\zeta^{i}_{BC}$ is  
\begin{align}\label{connectionform}
 \nabla_{a}\zeta^{i}_{BC} = \Gamma_{a}{}^{i}{}_{j}\zeta^{j}_{BC},
\end{align}
where the connection 1-form on the right side is $\Gamma_{a}{}^{i}{}_{j}=\zeta_{j}^{BC}\nabla_{a}\zeta^{i}_{BC}$, and $\zeta_{j}^{BC}=\zeta^{jBC}$. The non-trivial, independent components are 
\begin{subequations}
\begin{align}
\Gamma_{a}{}^{3}{}_{1} ={}& - (o^{B}\nabla_{a}o_{B} + \iota^{B}\nabla_{a}\iota_{B}), \label{Gamma31} \\
\Gamma_{a}{}^{3}{}_{2} ={}& i(o^{B}\nabla_{a}o_{B} - \iota^{B}\nabla_{a}\iota_{B}), \label{Gamma32} \\ 
\Gamma_{a}{}^{2}{}_{1} ={}& -2i \,\iota^{B}\nabla_{a}o_{B}. \label{Gamma21}
\end{align}
\end{subequations}

Our convention for the Riemann tensor is 
\begin{align}\label{Riemann}
(\nabla_{a}\nabla_{b} - \nabla_{b}\nabla_{a})v^{d} = - R_{abc}{}^{d}v^{c}.
\end{align} 
We stress that {\em this is the opposite of the convention used by Penrose and Rindler \cite{PR1, PR2}} (compare to \cite[Eq. (4.2.31)]{PR1}). This means that the Riemann tensor here is minus the Riemann tensor in \cite{PR1, PR2}, and similarly the curvature spinors here are minus the curvature spinors in \cite{PR1, PR2}. With the convention \eqref{Riemann}, we have for any spinors $\kappa_{C},\chi_{C'}$:
\begin{subequations}\label{commutators}
\begin{align}
 \Box_{AB}\kappa_{C} ={}& +\Psi_{ABC}{}^{D}\kappa_{D} + \frac{R}{24}(\epsilon_{AC}\kappa_{B} + \epsilon_{BC}\kappa_{A}), \\
 \Box_{A'B'}\chi_{C'} ={}& +\tilde\Psi_{A'B'C'}{}^{D'}\chi_{D'} + \frac{R}{24}(\epsilon_{A'C'}\chi_{B'} + \epsilon_{B'C'}\chi_{A'}),
\end{align}
\end{subequations}
where $\Box_{AB}=\nabla_{A'(A}\nabla^{A'}_{B)}$ and $\Box_{A'B'}=\nabla_{A(A'}\nabla^{A}_{B')}$ (compare to \cite[Eq. (4.9.15)]{PR1}). In \eqref{commutators}, $R=g^{ac}R_{abc}{}^{b}$ is the scalar curvature, and $\Psi_{ABCD}$ and $\tilde\Psi_{A'B'C'D'}$ are the Weyl curvature spinors. The Weyl tensor is $W_{abcd} = W^{+}_{abcd} + W^{-}_{abcd}$, where
\begin{align}
W^{+}_{abcd} = \Psi_{ABCD}\epsilon_{A'B'}\epsilon_{C'D'}, \qquad 
W^{-}_{abcd} =\tilde\Psi_{A'B'C'D'}\epsilon_{AB}\epsilon_{CD}
\end{align}
are, respectively, the self-dual and anti-self-dual parts (according to our conventions; cf. below \eqref{SD2forms}).

\section{Charges and gravitational perturbations}
\label{sec:conservationlaw}

\subsection{Conformal geometry}
\label{sec:CGHP}

The calculations in this section are most easily performed using a conformally invariant GHP connection $\mathcal{C}_{a}$ \cite{Araneda:2021wcd}. Here we summarize the necessary background. 

By a conformal transformation we mean a rescaling of the metric
\begin{align}\label{conftransf}
 g_{ab} \to \hat{g}_{ab} = \Omega^{2}g_{ab},
\end{align}
where $\Omega$ is a positive scalar field. Consider one of the real self-dual 2-forms in \eqref{SD2forms}, say $Z_{ab}\equiv Z^{3}_{ab}$. Then $J^{a}{}_{b}=\sqrt{2}Z_{bc}g^{ca}$ is an almost-complex structure compatible with $g_{ab}$. Under a conformal change \eqref{conftransf}, $\hat{Z}_{ab} = \Omega^{2}Z_{ab}$, so $J^{a}{}_{b}$ is conformally invariant.
The Lee form $f=f_{a}\d{x}^a$ (see e.g. \cite{Gover}) is defined by 
\begin{align}\label{idLeeform}
 \d{Z} = -2 f\wedge Z.
\end{align}
(In indices, $\partial_{[a}Z_{bc]}=-2f_{[a}Z_{bc]}$.)
The principal spinors of $Z_{ab}$ are $o_A,\iota_A$. Note that $J^{a}{}_{b}$ is invariant under the ``GHP'' transformation
\begin{align}\label{GHPtr}
o_A \to z o_A, \qquad \iota_A \to z^{-1}\iota_A, 
\end{align}
where $z\in U(1)$. The almost-complex structure is also invariant under a conformal transformation
\begin{align}\label{CTspinors} 
g_{ab} \to \Omega^{2}g_{ab}, \qquad
o_A\to\Omega^{1/2}o_A, \qquad \iota_A\to\Omega^{1/2}\iota_A.
\end{align}
A spinor/tensor field $\varphi^{A...A'...}_{B...B'...}$ which transforms under \eqref{GHPtr}-\eqref{CTspinors} as $\varphi^{A...A'...}_{B...B'...} \to \Omega^{w}z^{p}\varphi^{A...A'...}_{B...B'...}$ is said to have conformal weight $w$ and GHP weight $p$. These `weighted' fields $\varphi^{A...A'...}_{B...B'...}$ are to be thought of as sections of vector bundles over $M$. A linear connection $\mathcal{C}_{AA'}$ on these bundles was constructed in \cite{Araneda:2021wcd}. If $\phi_{AB...K}$ is a totally symmetric spinor field with conformal weight $w$ and GHP weight $p$, we have:
\begin{equation}\label{defC}
\begin{aligned}
\mathcal{C}_{AA'}\phi_{BC...L} ={}& \nabla_{AA'}\phi_{BC...L} + w f_{AA'}\phi_{BC...L} + p P_{AA'}\phi_{BC...L} \\
& - f_{BA'}\phi_{AC...L} - f_{CA'}\phi_{BA...L} - ... - f_{LA'}\phi_{BC...A},
\end{aligned}
\end{equation}
where the 1-forms $f_{a}$ and $P_{a}$ are
\begin{align}\label{fP}
 f_{a} = -\frac{1}{2} J^{c}{}_{b}\nabla_{c}J^{b}{}_{a}, 
 \qquad P_a = \frac{i}{2}(\Gamma_{a}{}^{2}{}_{1} - J^{b}{}_{a}f_{b}),
\end{align}
and $\Gamma_{a}{}^{2}{}_{1}$ was defined in \eqref{Gamma21}. The 1-form $f_{a}$ is the Lee form, cf. \eqref{idLeeform}. 
\begin{remark}\label{remark:CGHP}
We shall later need the following facts:
\begin{enumerate}
\item If $\phi_{AB...K}$ has $w=-1$, $p=0$, then
\begin{align}
\mathcal{C}_{A'}{}^{B}\phi_{AB...K} = \nabla_{A'}{}^{B}\phi_{AB...K}, \label{cwzero}
\end{align}
whereas for $w=0$, $p=0$:
\begin{align}
\mathcal{C}_{A'}{}^{B}\phi_{AB...K} = \nabla_{A'}{}^{B}\phi_{AB...K}+f_{A'}{}^{B}\phi_{AB...K}. \label{cwminusone}
\end{align}
\item $o_A$ and $\iota_A$ have weights $(w=\frac{1}{2},p=1)$ and $(w=\frac{1}{2},p=-1)$ respectively, so a calculation using \eqref{defC} gives:
\begin{subequations}\label{Cdyad}
\begin{align}
 & \mathcal{C}_{AA'}o_{B} = \sigma^{0}_{A'}\iota_A\iota_B, \qquad 
 \sigma^{0}_{A'} = o^Ao^B\nabla_{AA'}o_B, \\
 & \mathcal{C}_{AA'}\iota_{B} = \sigma^{1}_{A'}o_Ao_B, \qquad
 \sigma^{1}_{A'} = \iota^A\iota^B\nabla_{AA'}\iota_{B}.
\end{align}
\end{subequations}
Note that $\sigma^{1}_{A'}=(\sigma^{0}_{A'})^{\dagger}$.
\item\label{itemHermitian} $(M,g_{ab})$ is Hermitian with respect to $J^{a}{}_{b}$ if and only if $\sigma^{0}_{A'}=0$. 
\end{enumerate}
\end{remark}

If $(M,g_{ab})$ is Einstein\footnote{All of our results for Einstein manifolds apply in particular in the Ricci-flat case. We shall not mention this explicitly in the following.} and Hermitian with respect to $J^{a}{}_{b}$, with $W^{+}_{abcd}\neq0$, then the Goldberg-Sachs theorem implies that $Z_{ab}$ is an eigenform of $W^{+}{}_{ab}{}^{cd}$:
\begin{align}\label{backgroundcurvature}
 \tfrac{1}{2}W^{+}{}_{ab}{}^{cd} Z_{cd} = -2\Psi_2 Z_{ab},
\end{align}
where $\Psi_2 = -\frac{1}{8}W^{+}_{abcd}Z^{ab}Z^{cd}$. In this case, Bianchi identities imply that
\begin{align}\label{LeeformEH}
 f_{a} = \Psi_{2}^{-1/3} \nabla_{a} \Psi_{2}^{1/3}.
\end{align}
Since $\hat{g}_{ab}$ is also Hermitian, and the Lee form is \eqref{LeeformEH},
it follows that $\hat{g}_{ab}=\Psi_2^{2/3} g_{ab}$ is K\"ahler, and, defining $\omega_{ab}=\sqrt{2}Z_{ab}$, the (closed) K\"ahler form is 
\begin{align}\label{kahlerform}
\hat\omega_{ab} = \Psi_2^{2/3}\omega_{ab}.
\end{align}

\subsection{Charges for Hermitian instantons}
\label{sec:IdentityForCharges}

We shall extensively use the self-dual 2-form $Z^{3}_{ab}$ and spinor field $\zeta^{3}_{AB}$ in \eqref{Z3}, so it is convenient to denote
\begin{align}\label{omega}
 \zeta_{AB} \equiv \zeta^{3}_{AB}, \qquad Z_{ab} \equiv Z^{3}_{ab}.
\end{align}
\begin{prop}\label{prop:identityforms}
Let $(M,g_{ab})$ be a Riemannian manifold. Let $Z_{ab}=\zeta_{AB}\epsilon_{A'B'}$ be a real self-dual 2-form, with $\zeta_{AB}=-i\sqrt{2}o_{(A}\iota_{B)}$, and let $f_{AA'}$ be the associated Lee form as in \eqref{idLeeform}. Define 
\begin{align}
 \psi_{AB} = \Psi_{AB}{}^{CD}\zeta_{CD}.
\end{align}
Then the following identity holds:
\begin{align}\label{identityLMspinors}
 \nabla_{A'}{}^{B}\psi_{AB} = f_{A'}{}^{B}\psi_{AB} + S_{AA'} + K_{AA'}, 
\end{align}
where 
\begin{align}
 S_{AA'} ={}& \zeta^{CD}\nabla_{A'}{}^{B}\Psi_{ABCD}, \label{bianchiterm} \\
 K_{AA'} ={}& i\sqrt{2} \, \Psi_{ABCD}(\sigma^{0}_{A'}\iota^{B}\iota^{C}\iota^{D}+\sigma^{1}_{A'}o^{B}o^{C}o^{D}), \label{quadratic}
\end{align}
and $\sigma^{0}_{A'}, \sigma^{1}_{A'}$ are defined in \eqref{Cdyad}.
\end{prop}

\begin{remark}
The tensor version of \eqref{identityLMspinors} is 
\begin{align}\label{identityLMtensors}
 \d \Sigma = f \wedge \Sigma + *S + *K.
\end{align}
where we defined the self-dual 2-form
\begin{align}\label{mass2form}
 \Sigma_{ab} :=  \Psi_{AB}{}^{CD}\zeta_{CD}\epsilon_{A'B'} 
 = \tfrac{1}{2}W^{+}{}_{ab}{}^{cd}Z_{cd}.
\end{align}
\end{remark}

\begin{proof}
The proof of \eqref{identityLMspinors} is straightforward if one uses the conformal-GHP connection from section \ref{sec:CGHP}.
Since $\zeta_{AB}$ has conformal weight $w=1$, then $\zeta^{AB}$ has $w=-1$, and since $\Psi_{ABCD}$ has $w=0$, then $\psi_{AB}$ has $w=-1$. Thus:
\begin{align*}
 \nabla_{A'}{}^{B}\psi_{AB} ={}& \mathcal{C}_{A'}{}^{B}\psi_{AB} \\
 ={}& (\mathcal{C}_{A'}{}^{B}\Psi_{ABCD})\zeta^{CD} 
 + \Psi_{ABCD}\mathcal{C}_{A'}{}^{B}\zeta^{CD} \\
 ={}& (\nabla_{A'}{}^{B}\Psi_{ABCD}+f_{A'}{}^{B}\Psi_{ABCD})\zeta^{CD} 
 + i\sqrt{2} \, \Psi_{ABCD}\mathcal{C}_{A'}{}^{B}(o^C\iota^D) \\
 ={}& f_{A'}{}^{B}\psi_{AB} +\zeta^{CD}\nabla_{A'}{}^{B}\Psi_{ABCD}+ K_{AA'},
\end{align*}
where in the first line we used \eqref{cwzero}, in the third we used \eqref{cwminusone}, and in the fourth we used \eqref{Cdyad} and the definition \eqref{quadratic}.
\end{proof}

If $(M,g_{ab})$ is Einstein, then the Weyl spinor satisfies 
\begin{align}\label{DivFreeWeyl}
 \nabla^{AA'}\Psi_{ABCD} = 0,
\end{align}
so the term with $S$ in \eqref{identityLMtensors} vanishes. Suppose that, in addition, $(M,g_{ab})$ is Hermitian, with complex structure $J^{a}{}_{b}=\sqrt{2}Z_{bc}g^{ca}$ (so that the principal spinors are $o_A, \iota_A$, recall \eqref{omega} and \eqref{Z3}). Then from Remark \ref{remark:CGHP}, we have $\sigma^{0}_{A'}=0$ and the term with $*K$ in \eqref{identityLMtensors} vanishes. Recalling that for an Einstein-Hermitian geometry with $\Psi_2\neq0$, the Lee form is \eqref{LeeformEH}, the result \eqref{identityLMtensors} reduces to 
\begin{align}\label{closed2f}
\d\left[ \Psi_2^{-1/3}\Sigma \right] = 0.
\end{align}

One way of interpreting this identity is in terms of spin-lowering via Killing spinors \cite[§6.4]{PR2}. We have $\Psi_2^{-1/3}\Sigma_{ab} \propto \Psi_{ABCD}K^{CD}\epsilon_{A'B'}$, where $K_{AB} = \Psi_2^{-1/3} o_{(A}\iota_{B)}$ is a valence-2 Killing spinor, $\nabla_{A'(A}K_{BC)}=0$ (cf. \cite[Eq. (6.7.15)]{PR2}). Then \eqref{closed2f} is the statement that 
\begin{align}\label{charge2form}
 \Psi_{ABCD}K^{CD}\epsilon_{A'B'}
\end{align}
is an anti-self-dual Maxwell field. As explained by Penrose and Rindler in \cite[§6.4]{PR2}, 
this construction can be applied to linearized gravity on Minkowski space-time by replacing $\Psi_{ABCD}$ with a spin 2 zero-rest-mass field. Then one can associate a charge to this field by integrating \eqref{charge2form} on a closed 2-surface, where, since Minkowski space-time is topologically trivial, the 2-surface must surround a region of matter sources in order for the charge to be non-vanishing. 

For our Einstein Hermitian geometry $(M,g_{ab})$, an analogous non-vanishing charge requires the existence of non-contractible 2-surfaces in $M$. Additionally, we can also interpret \eqref{closed2f} in terms of complex structures: 
using \eqref{mass2form} and \eqref{backgroundcurvature}, 
eq. \eqref{closed2f} is a reflection of the fact that $(M,g_{ab})$ is conformally K\"ahler, with the K\"ahler form given by \eqref{kahlerform}.
Combining the two interpretations above, we can then define a charge associated to every 2-cycle $S$ in $M$, by the formula \eqref{charge}.

\subsection{Linearized charges}
\label{sec:perturbationsLM}

Given a smooth one-parameter family of tensor fields $T(s)$ in $M$, we define the $k$-jets
\begin{align}
\delta^{k} T = \left. \left[\frac{\d^{k} }{\d s^{k}}T(s)\right] \right|_{s=0}.
\end{align}
The zeroth jet $T= T(0)$ is the background value of $T(s)$, and the first jet $\delta{T}$ is the linearization of $T(s)$ at $T$. We recall that, under an infinitesimal diffeomorphism generated by an arbitrary vector field $\xi^a$, the linearization $\delta{T}$ transforms as 
\begin{align}\label{gaugetr}
 \delta{T} \to \delta{T} + \pounds_{\xi}T,
\end{align}
which in this context is referred to as a gauge transformation. 
We say that the linearization $\delta{T}$ is gauge invariant if it is invariant under \eqref{gaugetr} for all $\xi^{a}$. This holds, in particular, if $T=0$.

\begin{thm}\label{thm:closed2form}
Let $g_{ab}(s)$ be a one-parameter family of Riemannian metrics in $M$, such that, for all $s$,
\begin{align}
 \nabla^{(s)\, a}W^{+}_{abcd}(s)=0. \label{divweyl}
\end{align}
Let $Z_{ab}(s)$ be a real self-dual 2-form, $Z_{ab}(s)Z^{ab}(s)=2$, and define 
\begin{align}
\Sigma_{ab}(s) :={}& \frac{1}{2}W^{+}{}_{ab}{}^{cd}(s)Z_{cd}(s), \label{Sigmas} \\
\Psi_{2}(s) :={}& -\frac{1}{8}W^{+}_{abcd}(s)Z^{ab}(s)Z^{cd}(s), \\
\hat\omega_{ab}(s) :={}& -\frac{\sqrt{2}}{3} \left[\Psi_2(s)^{-1/3}\Sigma_{ab}(s) - \Psi_2(s)^{2/3} Z_{ab}(s) \right]. \label{KFepsilon}
\end{align}
Suppose the background metric $g_{ab}=g_{ab}(0)$ is Einstein and Hermitian, with (closed) self-dual K\"ahler form $\hat\omega_{ab}=\hat\omega_{ab}(0)$, and $\Psi_2=\Psi_2(0)\neq0$. Then:
\begin{align}\label{pertclosedform}
 \d \, \delta\hat\omega =0,
\end{align}
where $\delta\hat\omega$ is the 2-form
\begin{align}\label{pertclosedform2}
\delta\hat\omega_{ab} = -\frac{\sqrt{2}}{3}\left[ \Psi_2^{-1/3}\delta\Sigma_{ab} - \Psi_2^{2/3}\delta{Z}_{ab} \right].
\end{align}
Moreover, the result \eqref{pertclosedform} is gauge invariant.
\end{thm}

\begin{proof}
For all $s$, we have
\begin{align}
\d{Z}(s) ={}& -2 f(s)\wedge Z(s), \label{dSigmae}  \\
\d\Sigma(s) ={}& f(s)\wedge\Sigma(s) + (*K)(s), \label{identityLMtensors1}
\end{align}
where the first equality is simply the definition of the Lee form (recall \eqref{idLeeform}), and the second equality follows from \eqref{identityLMtensors} after using that \eqref{divweyl} holds for all $s$. Since the background is Einstein and Hermitian, then $(*K)(s)=O(s^2)$, so $\d(*K)/\d s \, |_{s=0}=0$. Thus, the linearization of \eqref{dSigmae}-\eqref{identityLMtensors1} is 
\begin{align*}
\d\, \delta Z ={}& -2\delta{f}\wedge Z - 2 f \wedge\delta Z, \\
\d\, \delta\Sigma ={}& \delta{f}\wedge\Sigma + f \wedge\delta\Sigma,
\end{align*}
where $Z=Z(0)$, $\Sigma=\Sigma(0)$ and $f=f(0)$. Using the background identities \eqref{LeeformEH} and \eqref{backgroundcurvature}, it follows that 
\begin{align*}
 \delta{f}\wedge\Sigma = -2\Psi_2\delta{f}\wedge Z 
 = \Psi_2(\d\delta Z+2f\wedge\delta Z) 
 = \Psi_2^{1/3}\d\left[ \Psi_2^{2/3} \delta Z \right],
\end{align*}
so
\begin{align*}
\d\delta\Sigma = \Psi_2^{1/3}\d\left[ \Psi_2^{2/3} \delta Z \right] 
+ \Psi_{2}^{-1/3}\d\Psi_2^{1/3} \wedge\delta\Sigma
\end{align*}
and therefore
\begin{align*}
\Psi_2^{1/3}\d\left[\Psi_2^{-1/3}\delta\Sigma\right] = \Psi_2^{1/3}\d\left[ \Psi_2^{2/3} \delta Z \right], 
\end{align*}
which gives \eqref{pertclosedform}.

Finally, consider a gauge transformation $\delta\hat\omega \to \delta\hat\omega + \pounds_{\xi}\hat\omega$. We have $\pounds_{\xi}\hat\omega = \d(\xi\lrcorner\,\hat\omega)$ since $\hat\omega$ is closed. Hence, $\delta\hat\omega \to \delta\hat\omega +  \d( \xi \lrcorner\,\hat\omega)$, so the result \eqref{pertclosedform} is gauge-invariant.
\end{proof}

\section{Examples} 
\label{sec:Examples}

In this section we shall analyze several examples of the charges \eqref{charge}, by computing the conformal K\"ahler structures, the 2-cycles, and then the periods \eqref{charge}, for both compact and non-compact instantons. For each of these examples, Theorem \ref{thm:closed2form} implies that generic gravitational perturbations also have a quasi-locally conserved charge, given by \eqref{perturbedcharge}.

A non-trivial issue is how to find the independent 2-cycles, for which we need the second homology group $H_2(M)$. To do this, we shall use the fact that all of the examples are toric. This implies that one can apply the rod structure formalism as in \cite{Chen:2010zu}, which allows to determine the topological properties of the instanton. For example, for a rod structure with $n$ turning points, the Euler characteristic $\chi(M)$ is equal to $n$. In the non-compact case, we see from \cite[Theorem 1]{Nilsson:2023ina} that $H_2(M)\cong\mathbb{Z}^{n-1}$.
In the compact case, we can use the fact that $\chi(M)$ is also the alternating sum of the Betti numbers $b_i(M)$ to deduce the number of independent 2-cycles.

Below we shall also use the terminology of `nuts' and `bolts', for which we refer to the paper by Gibbons and Hawking \cite{Gibbons:1979xm}.

\subsection{Ambi-K\"ahler examples}
The examples in this subsection are conformally K\"ahler for both orientations. In \cite{Apostolov:2013oza}, these geometries are referred to as {\em ambi-K\"ahler}. The complex structures are $\omega^{+}_{ab}=2io_{(A}\iota_{B)}\epsilon_{A'B'}$ and $\omega^{-}_{ab}=2i\tilde{o}_{(A'}\tilde\iota_{B')}\epsilon_{AB}$, and the K\"ahler forms are $\hat\omega^{\pm}_{ab}$. We denote the only non-trivial components of $W^{\pm}_{abcd}$ by $\Psi_{2}^{\pm}$. (Note $\Psi_{2}^{+}=\Psi_2$.)

\begin{example}[Kerr]\label{example:Kerr}
The AF Kerr instanton with parameters $m>0,a$ has the following local form of the metric: 
\begin{align}
 g ={}& \frac{\Delta}{\Sigma}(\d\tau-a\sin^2\theta\d\phi)^2+\frac{\sin^2\theta}{\Sigma}[a\d\tau+(r^2-a^2)\d\phi]^2+\frac{\Sigma}{\Delta}\d{r}^2+\Sigma\d\theta^2, \label{Kerr}
\end{align}
where $\Delta=r^2 - 2mr -a^2$, $\Sigma=r^2-a^2\cos^2\theta$. Both $W^{\pm}_{abcd}$ are type D, with 
\begin{align}\label{psi2kerr}
 \Psi_{2}^{\pm} = \frac{m}{(r \pm a\cos\theta)^3}.
\end{align}
The following 2-forms 
\begin{align}
 \omega^{\pm} = (\d\tau-a\sin^2\theta\d\phi)\wedge\d{r} 
  \pm \sin\theta\d\theta\wedge[a\d\tau+(r^2-a^2)\d\phi]
 \label{CSKerr}
\end{align}
define integrable almost-complex structures with opposite orientation. The K\"ahler forms are then given by $\hat\omega^{\pm}=(\Psi_2^{\pm})^{2/3}\omega^{\pm}$. 
The topology is $\mathbb{R}^2\times S^2$, and the rod structure has two turning points, two semi-infinite rods and one finite rod \cite[§4.4]{Chen:2010zu}. 
Since there are two turning points, we have $H_2(M)\cong\mathbb{Z}$ \cite[Thm. 1]{Nilsson:2023ina}, so there is only one independent 2-cycle $S$, corresponding to the bolt at $r=r_{+}$ of the Killing field $\partial_{\tau}+\Omega_{E}\partial_{\phi}$, where $r_{+}=m+\sqrt{m^2+a^2}$ and $\Omega_{E}=a/(2mr_{+})$. The associated charges \eqref{charge} are 
\begin{align}\label{chargeKerr}
\int_{S}\hat\omega^{\pm} =  \pm \, 4\pi \, m^{2/3}.
\end{align}
\end{example}

For our next two examples, we first recall that the Ricci-flat Taub-NUT metric, with parameters $m>0,n>0$ is
\begin{align}
g ={}& 4n^2\frac{\Delta}{\Sigma}(\d\psi + \cos\theta\d\phi)^2+\frac{\Sigma}{\Delta}\d{r}^2 + \Sigma(\d\theta^2+\sin^2\theta\d\phi^2), \label{TN} 
\end{align}
where $\Sigma=r^2-n^2$, $\Delta = r^2-2mr+n^2$. Both $W^{\pm}_{abcd}$ are generically type D,
\begin{align}\label{Psi2TNlambda}
 \Psi_{2}^{\pm} = \frac{m \pm n}{(r \pm n)^{3}}.
\end{align}
The 2-forms
\begin{align}
 \omega^{\pm} = 2n(\d\psi + \cos\theta\d\phi)\wedge\d{r} 
\pm \Sigma\sin\theta\d\theta\wedge\d\phi
\label{CSTN}
\end{align}
correspond to
integrable almost-complex structures $J^{\pm}$ with opposite orientation ($(J_{\pm})^{a}{}_{b}=\omega^{\pm}_{bc}g^{ca}$). These structures are {\em not} K\"ahler.

\begin{example}[Self-dual Taub-NUT]\label{example:SDTN}
The self-dual Taub-NUT instanton is obtained from \eqref{TN} when $m=n$. It is ALF, and the topology is $S^4\setminus\{\rm pt.\} \cong \mathbb{R}^4$. From \eqref{Psi2TNlambda}, we see that $\Psi^{-}_2=0$ and $\Psi^{+}_2=2n/(r+n)^{3}$, that is $W^{-}_{abcd}=0$ and $W^{+}_{abcd}\neq0$, hence the space is hyper-K\"ahler with respect to one orientation, i.e. there is a 2-sphere of K\"ahler structures, and type D with respect to the other. 
The 2-forms \eqref{CSTN} are however not K\"ahler; the geometry is not only hyper-K\"ahler but also strictly conformally K\"ahler w.r.t. both orientations \cite{a23}. The rod structure has one turning point and two semi-infinite rods \cite[§4.2]{Chen:2010zu}.
However, since the manifold is topologically trivial, there are no 2-cycles; thus the parameter $n$ does not correspond to a charge integral like \eqref{charge}. For a different construction involving 3-cycles, see \cite{Bossard:2008sw}.
\end{example}

\begin{example}[Taub-bolt]\label{example:Taub-bolt}
The Taub-bolt instanton is obtained from \eqref{TN} when $m=\frac{5}{4}n$. It is ALF, and the topology is $\mathbb{CP}^2\setminus\{ {\rm pt.} \}$. We have $\Psi_2^{-}=n/[4(r-n)^3]$, $\Psi_2^{+}=9n/[4(r+n)^3]$, so both $W^{\pm}_{abcd}$ remain type D. The rod structure has two turning points, two semi-infinite rods and one finite rod \cite[§4.7]{Chen:2010zu}. 
As in Kerr we have two turning points, so $H_2(M)\cong\mathbb{Z}$ \cite[Thm. 1]{Nilsson:2023ina} and there is one independent 2-cycle, 
which we denote by $S_{2n}$ and corresponds to the bolt of $\partial_{\psi}$ at $r=2n$. We then find:
\begin{align}
\int_{S_{2n}}\hat\omega^{\pm} = c^{\pm}\, n^{2/3}
\end{align}
where $c^{-}=-3\cdot 4^{1/3}\pi$, $c^{+}=3^{1/3}\cdot 4^{1/3}\pi$.
\end{example}

Examples \ref{example:Kerr}, \ref{example:SDTN} and \ref{example:Taub-bolt} are AF/ALF. Our construction also applies to the ALE case:

\begin{example}[Eguchi-Hanson]
The metric for the Eguchi-Hanson instanton is locally given by
\begin{align}\label{EH}
 g = f(r)\frac{r^2}{4}(\d\psi+\cos\theta\d\phi)^2 + \frac{\d r^2}{f(r)} + \frac{r^2}{4}(\d\theta^2+\sin^2\theta\d\phi^2), 
\end{align}
where $f(r)=1-(a/r)^4$, and $a$ is a real constant. The topology is $T^{*}\mathbb{CP}^1$, and the asymptotics is ALE. The metric is Ricci-flat, the self-dual Weyl tensor is $W^{+}_{abcd}=0$, and the anti-self-dual Weyl tensor $W^{-}_{abcd}$  is type D \footnote{As noted by Gibbons and Pope \cite{Gibbons:1979xn}, the Eguchi-Hanson metric \eqref{EH} can be obtained from \eqref{TN} by setting $m=n+a^4/(128 n^3)$, $r=m+\tilde{r}^2/(8n)$, and taking the limit $n\to\infty$ (and then redefining $\tilde{r}\to r$), see \cite[Eqs. (4.15)-(4.16)]{Gibbons:1979xn}. This procedure leads to \eqref{psi2EH} (i.e., $W^{-}_{abcd}\neq0$ and $W^{+}_{abcd}=0$, which is the opposite to self-dual Taub-NUT).}: 
\begin{align}\label{psi2EH}
 \Psi_2^{-} = \frac{a^4}{4 r^6}, \qquad \Psi_2^{+} = 0.
\end{align}
One can check that the following 2-forms
\begin{align*}
\omega^{\pm} = \frac{r}{2}(\d\psi+\cos\theta\d\phi)\wedge\d r 
\pm \frac{r^2}{4}\sin\theta\d\theta\wedge\d\phi
\end{align*}
give integrable almost-complex structures $J^{\pm}$ with opposite orientation, and that $\d\omega^{+}=0$, so $J^{+}$ is K\"ahler. However $\d\omega^{-}\neq0$, and this gives the conformally K\"ahler form $\hat\omega^{-}=(\Psi_2^{-})^{2/3}\omega^{-}$ needed to define the charge \eqref{charge}. 
The rod structure has two turning points, two semi-infinite rods and one finite rod \cite[§4.5]{Chen:2010zu}. 
As in previous cases, two turning points give $H_2(M)\cong\mathbb{Z}$ \cite[Thm. 1]{Nilsson:2023ina} so there is one independent 2-cycle, which we denote by $S_{a}$ and corresponds to the bolt of $\partial_{\psi}$ at $r=a$. We have
\begin{align}
 \int_{S_{a}}\hat\omega^{-} = -\pi\left(\tfrac{a}{4}\right)^{2/3} .
\end{align}
\end{example}

So far the examples are Ricci-flat, but we can also include a non-trivial cosmological constant $\lambda\neq0$. For $\lambda>0$ the instanton is necessarily compact, and Einstein-Hermitian (non-K\"ahler) compact 4-manifolds with $\lambda>0$ have been classified \cite{Lebrun2010}: these are the $\mathbb{CP}^2$ instanton \cite{Gibbons:1978zy} (with orientation opposite to the K\"ahler one), the Page metric on $\mathbb{CP}^2\#\overline{\mathbb{CP}}{}^2$ \cite{Page:1978vqj}, or the Chen-LeBrun-Weber metric on $\mathbb{CP}^2\#2\overline{\mathbb{CP}}{}^2$ \cite{ChenLebrunWeber}. Only the first two metrics are explicitly known.

\begin{example}[$\mathbb{CP}^2$]
The $\mathbb{CP}^2$ instanton \cite{Gibbons:1978zy} is 
\begin{align}\label{metricCP2}
g = \frac{r^2}{4f(r)^2}(\d\psi+\cos\theta\d\phi)^2 + \frac{\d{r}^2}{f(r)^2} 
+ \frac{r^2}{4f(r)}(\d\theta^2+\sin^2\theta\d\phi^2),
\end{align}
where $f(r)=1+\frac{\lambda}{6}r^2$. The metric is Einstein, with cosmological constant $\lambda$. One can check that the following 2-forms
\begin{align*}
\omega^{\pm} = \frac{r}{2f^2}(\d\psi+\cos\theta\d\phi)\wedge\d{r} 
\pm \frac{r^2}{4f}\sin\theta\d\theta\wedge\d\phi
\end{align*}
give integrable almost-complex structures with opposite orientation. Furthermore, $\d\omega^{+}=0$, $\d\omega^{-}\neq0$, and $\d\hat\omega^{-}=0$, where $\hat\omega^{-}=\frac{4f^2}{r^4}\omega^{-}$. Thus, \eqref{metricCP2} is K\"ahler on the self-dual side, and strictly conformally K\"ahler on the anti-self-dual side. We also have 
\begin{align}\label{psi2CP2}
 \Psi_2^{-} = 0, \qquad \Psi_2^{+} = \frac{\lambda}{3}.
\end{align}
The rod structure has three turning points \cite{Chen:2010zu}, at $r=0$ and at 
$(r,\theta)=(\infty,0),(\infty,\pi)$. Thus the Euler characteristic is $\chi=3$, and since the manifold is compact, this gives $b_2(M)=1$, so there is only one 2-cycle, corresponding to the set $r=\infty$, $0<\theta<\pi$, which is a bolt of $\partial_{\psi}$ that we denote $S_{\infty}$. Since $\d\hat\omega^{-}=0$, a charge like \eqref{charge} can be associated to $\hat\omega^{-}$; however, $\Psi_2^{-}=0$ implies that this would not be related to the gravitational field. For $\hat\omega^{+}=(\Psi_2^{+})^{2/3}\omega^{+}$, the charge is 
\begin{align}
\int_{S_{\infty}}\hat\omega^{+} = (2\cdot 3^{1/3} \pi) \, \lambda^{1/3}.
\end{align}
\end{example}

\begin{example}[Page]
The Page metric \cite{Page:1978vqj, Gibbons:1979xm, Page:2009dm} is given by 
\begin{equation}
\begin{aligned} 
g = \frac{3(1+\nu^2)}{\lambda} & \left[ 
\frac{f(r)\sin^2r}{4(3+\nu^2)^2}(\d\psi+\cos\theta\d\phi)^2
+ \frac{\d{r}^2}{f(r)}  \right. \\
& \left. + \frac{(1-\nu^2\cos^2r)}{(3+6\nu^2-\nu^4)}(\d\theta^2+\sin^2\theta\d\phi^2) \right]
\end{aligned}
\end{equation}
where $\nu$ is the positive root of $\nu^4 + 4\nu^3 - 6\nu^2 + 12\nu - 3 = 0$ ($\nu\approx 0.2817$), and
\begin{align*}
 f(r) = \frac{3-\nu^2-\nu^2(1+\nu^2)\cos^2r}{1-\nu^2\cos^2r}.
\end{align*}
The space is Einstein, with cosmological constant $\lambda$. 
Both Weyl tensors $W^{\pm}_{abcd}$ are type D, and a calculation shows that 
\begin{align}
 \Psi_{2}^{\pm} = \frac{\lambda}{3}\frac{(1-\nu^2)^2}{(1+\nu^2)}\frac{1}{(1 \pm \nu\cos r)^{3}}. 
\end{align}
One can check that the 2-forms
\begin{align*}
 \omega^{\pm} = \frac{3(1+\nu^2)}{\lambda(3+\nu^2)}\left[\frac{\sin r}{2}(\d\psi+\cos\theta\d\phi)\wedge\d r \pm \frac{(1-\nu^2\cos^2r)}{4\nu}\sin\theta\d\theta\wedge\d\phi \right]
\end{align*}
give oppositely oriented integrable almost-complex structures. The topology is $\mathbb{CP}^2\#\overline{\mathbb{CP}}{}^2$, and it corresponds to blowing up the nut of $\mathbb{CP}^2$ at $r=0$. 
The rod structure has four turning points \cite{Chen:2010zu} (at $(r,\theta)=(0,0),(0,\pi),(\pi,0),(\pi,\pi)$), so the Euler characteristic is 4.
This gives $b_2(M)=2$, hence there are two independent 2-cycles. 
These are represented by the two bolts of $\partial_{\psi}$ (cf. \cite{Page:2009dm}), which we denote by $S_{0},S_{\pi}$ and correspond to the sets $r=0,\pi$, $0<\theta<\pi$.
The charges \eqref{charge} are
\begin{align}
 \int_{S_{r_o}} \hat\omega^{\pm} = c^{\pm}_{r_o} \, \lambda^{-1/3}
\end{align}
where $r_o=0,\pi$, and 
\begin{align*}
c^{\pm}_{r_o} =  \pm 3^{1/3}\frac{(1+\nu^2)^{1/3}(1-\nu^2)^{4/3}}{\nu(3+\nu^2)}\frac{(1\mp\nu\cos r_o)}{(1\pm\nu\cos r_o)} \, \pi.
\end{align*}
\end{example}

\subsection{The Chen-Teo instanton}

The Chen-Teo family of ALF Ricci-flat metrics is \cite[Eq. (2.1)]{ChenTeo2}
\begin{align}\label{ChenTeo}
 g = \frac{(F\d\tau+G\d\phi)^2}{(x-y)HF} 
 + \frac{kH}{(x-y)^3}\left(\frac{\d{x}^2}{X} - \frac{\d{y}^2}{Y} 
 - \frac{XY}{kF}\d\phi^2 \right)
\end{align}
where $X=P(x)$, $Y=P(y)$, $F = y^2 X - x^2 Y$, and 
\begin{align}
\nonumber H ={}& (\nu x+y)[(\nu x -y)(a_1-a_3xy)-2(1-\nu)(a_0-a_4x^2y^2)], \\
\nonumber G = {}&(\nu^2a_0+2\nu a_3 y^3+a_4(2\nu-1)y^4)X + ((1-2\nu)a_0 -2\nu a_1 x-\nu^2 a_4 x^4)Y, \\
P(u) ={}& a_0 + a_1 u + a_2 u^2 + a_3 u^3 + a_4 u^4.
\label{ChenTeoFunctions}
\end{align}
Here, $a_0,...,a_4, k, \nu$ are real constants, with $\nu\in[-1,1]$. It was shown in \cite{Aksteiner:2021fae} that \eqref{ChenTeo} is one-sided type D, and thus conformally K\"ahler only on one side (which we take to be the self-dual side). The K\"ahler form is $\hat\omega = \Psi_2^{2/3} \omega$ (recall \eqref{kahlerform}), and to find an explicit expression we follow \cite{a24}. After replacing eq. (5.3) into eq. (2.10) in this reference, we get
\begin{equation}\label{CSChenTeo}
\begin{aligned}
 \omega = \frac{\sqrt{k}}{(x-y)^2} & \left\{ 
 (x\d{y}-y\d{x})\wedge\d\tau + \left(\frac{Gx}{F}+\frac{HyX}{F(x-y)} \right)\d{y}\wedge\d\phi \right. \\
 & \left. -\left(\frac{Gy}{F}+\frac{HxY}{F(x-y)} \right)\d{x} \wedge\d\phi \right\}
\end{aligned}
\end{equation}
The scalar $\Psi_2$ can be computed using \cite[Eq. (2.28)]{a24}; we find
\begin{align}\label{psi2ChenTeo}
 \Psi_2 = -\frac{(\nu+1)}{2k}\left(\frac{x-y}{\nu x+y}\right)^{3},
\end{align}
which coincides with the expression deduced from \cite[Eqs. (3.28) and (3.39)]{Aksteiner:2021fae}. The closed K\"ahler form is then $\hat\omega = \Psi_2^{2/3} \omega$.

The vector fields $\partial_{\tau}, \partial_{\phi}$ are Killing, and the solution is toric. The rod structure has three turning points, two finite rods and two semi-infinite rods \cite{ChenTeo2, Aksteiner:2021fae, Kunduri:2021xiv}. These can be conveniently described in terms of the four roots, $r_1,r_2,r_3,r_4$, of the polynomial $P$ in \eqref{ChenTeoFunctions}. 
Following Chen and Teo \cite[§3.3]{ChenTeo2}, we have
\begin{itemize}
\item $\mathcal{R}_1$: $x = r_2$, $r_1<y<r_2$, direction $\ell_1=(\alpha_1,\beta_1)$,
\item $\mathcal{R}_2$: $r_2<x<r_3$, $y=r_1$, direction $\ell_2=(\alpha_2,\beta_2)$,
\item $\mathcal{R}_3$: $x=r_3$, $r_1<y<r_2$, direction $\ell_3=(\alpha_3,\beta_3)$,
\item $\mathcal{R}_4$: $r_2<x<r_3$, $y=r_2$, direction $\ell_4=(\alpha_4,\beta_4)$,
\end{itemize}
where the notation $\ell_i=(\alpha_i,\beta_i)$ means $\ell_i = \alpha_i\partial_{\tau}+\beta_i\partial_{\phi}$, and the expressions for $\alpha_i,\beta_i$ can be found in \cite[Eqs. (3.14)-(3.15)]{ChenTeo2}. It is convenient to adapt the Killing coordinates to the rod vectors $\ell_1,\ell_2$ (cf. \cite{Aksteiner:2021fae}). This is done by defining $(\tilde\tau,\tilde\phi)$ via $\tau = \alpha_2\tilde\tau+\alpha_1\tilde\phi$, $\phi=\beta_2\tilde\tau+\beta_1\tilde\phi$, so that $\partial_{\tilde\phi}=\ell_1$ and $\partial_{\tilde\tau}=\ell_2$. The closed K\"ahler form $\hat\omega=\Psi_2^{2/3}\omega$ is then
\begin{equation}\label{KahlerFormChenTeo}
\begin{aligned}
 \hat\omega = \left(\tfrac{\nu+1}{2k}\right)^{2/3}\frac{\sqrt{k}}{(\nu x+y)^2} & \left\{ 
 \frac{1}{\kappa_2}\left[ \left(K_2[1]+\frac{G}{F}\right)x + \frac{HyX}{F(x-y)} \right]\d{y} \wedge \d\tilde\tau \right. \\
 & - \frac{1}{\kappa_2}\left[ \left(K_2[1]+\frac{G}{F}\right) y + \frac{HxY}{F(x-y)} \right]\d{x}\wedge \d\tilde\tau \\
 & + \frac{1}{\kappa_{1}}\left[ \left(K_1[1]+\frac{G}{F}\right) x + \frac{HyX}{F(x-y)} \right]\d{y}\wedge \d\tilde\phi \\
 & \left. - \frac{1}{\kappa_1}\left[ \left(K_1[1]+\frac{G}{F}\right) y + \frac{HxY}{F(x-y)} \right]\d{x}\wedge \d\tilde\phi \right\},
\end{aligned}
\end{equation}
where we used the notation $\alpha_{i}=K_{i}[1]/\kappa_{i}$ and $\beta_{i}=1/\kappa_{i}$ following \cite[Eqs. (3.14)-(3.15)]{ChenTeo2}. 

The family of ALF metrics \eqref{ChenTeo} has, in general, conical singularities. Chen and Teo showed in \cite[§4.3]{ChenTeo2} that, by imposing certain constraints on the parameters, one can obtain a smooth, AF, regular 4-manifold. This is the Chen-Teo {\it instanton} \cite{ChenTeo1}, which corresponds to setting the total NUT charge to zero and requiring the rod vectors to have closed orbits with $2\pi$ period and to satisfy $\ell_1=\ell_4=\pm\ell_2\pm\ell_3$, cf. \cite[Eq. (4.13)]{ChenTeo2} and \cite{Aksteiner:2021fae}. The solution to these constraints is given in \cite[Eq. (4.16)]{ChenTeo2}, from where one deduces that
\begin{align}
\nonumber & \nu \equiv -2\xi^2, \\
\nonumber & r_1=-4\xi^3(1-\xi), \quad
r_2=-\xi(1-2\xi+2\xi^2), \quad 
r_3=1-2\xi, \quad
r_4=\infty, \\
\nonumber & a_4=0, \quad a_3=1, \quad a_2=-1+3\xi-2\xi^2+6\xi^3-4\xi^4, \\
\nonumber  & a_1=-\xi+4\xi^2-10\xi^3+20\xi^4-20\xi^5+16\xi^6-8\xi^7, \\
 & a_0=4(1-\xi)(2\xi-1)(1-2\xi+2\xi^2)\xi^4. \label{ChenTeoInstanton}
\end{align}
The resulting family of metrics has 2 parameters, $(\xi,k)$. The parametrization \eqref{ChenTeoInstanton} gives $\alpha_1=\alpha_4$, $\beta_1=\beta_4$, $\alpha_3=-(\alpha_1+\alpha_2)$ $\beta_3=-(\beta_1+\beta_2)$. The rod vectors are then $\ell_1=\ell_4=\partial_{\tilde\phi}$, $\ell_2=\partial_{\tilde\tau}$, $\ell_3 = -(\partial_{\tilde\phi}+\partial_{\tilde\tau})$.

The topology is $\mathbb{CP}^2\setminus S^1$. Since there are three turning points, we have $H_2(M)\cong\mathbb{Z}^2$ \cite[Thm. 1]{Nilsson:2023ina}, so there are two independent 2-cycles, which we can take to be the finite rods $\mathcal{R}_2,\mathcal{R}_3$. Using the expression \eqref{KahlerFormChenTeo}, it is straightforward to integrate $\hat\omega$ over $\mathcal{R}_2,\mathcal{R}_3$.
Since the rods correspond to roots of $X,Y$, the terms with $H$ in \eqref{KahlerFormChenTeo} vanish on the rods, and also (from \eqref{ChenTeoFunctions}) $G/F$ becomes constant, so we find (using $a_4=0$, $a_3=1$)
\begin{align*}
\hat\omega \big|_{\mathcal{R}_2} ={}& -\left(\tfrac{\nu+1}{2k}\right)^{2/3}\tfrac{\sqrt{k}}{\kappa_1 r_1} \left( K_1[1] r_1^2 + \nu^2 a_0 + 2\nu r_1^3 \right)\frac{\d{x}\wedge\d\tilde\phi}{(\nu x + r_1)^2}, \\
\nonumber \hat\omega \big|_{\mathcal{R}_3} ={}& \left(\tfrac{\nu+1}{2k}\right)^{2/3} \tfrac{\sqrt{k}}{r_3} \left[ \left(\tfrac{K_2[1]}{\kappa_2} - \tfrac{K_{1}[1]}{\kappa_1}\right) r_3^2  \right. \\
& \left. 
- (\tfrac{1}{\kappa_2}-\tfrac{1}{\kappa_1})[(1-2\nu)a_0 - 2\nu a_1r_3 ] \right] \frac{\d{y}\wedge\d\tilde\tau}{(\nu r_3+y)^2}
\end{align*}
where for $\hat\omega \big|_{\mathcal{R}_3}$ we used the fact that $\tilde\tau+\tilde\phi$ is constant on $\mathcal{R}_3$. Integrating and replacing \eqref{ChenTeoInstanton}, we find the charges
\begin{align}
 \int_{\mathcal{R}_2}\hat\omega ={}& \frac{4\pi}{2^{2/3}}\left(\frac{k}{1-2\xi^2} \right)^{1/3}, \\
 \int_{\mathcal{R}_3}\hat\omega ={}& -\frac{16\pi}{2^{2/3}}\left(\frac{k}{1-2\xi^2} \right)^{1/3}\xi^2.
\end{align}
\begin{remark}\label{remark:twoparameters}
This result implies that {\em both} of the two parameters of the Chen-Teo instanton can be obtained as charge integrals of the K\"ahler form. As mentioned in the introduction, this is in contrast to the Kerr case, where the K\"ahler construction gives only one of the two parameters, namely the mass, see \eqref{chargeKerr}. We can understand this difference from the fact that there are two bolts in Chen-Teo, whereas Kerr has only one. 
In a sense, the Chen-Teo instanton is a ``double'' object, as suggested by the interpretation in \cite[§5]{ChenTeo2} in terms of two touching Kerr-NUTs.
\end{remark}

\subsection*{Acknowledgements}
The authors would like to thank the Isaac Newton Institute for Mathematical Sciences, Cambridge, for support and hospitality during the programme Twistor Theory where work on this paper was undertaken. This work was supported by EPSRC grant no EP/R014604/1. BA is supported by the ERC Consolidator/UKRI Frontier grant TwistorQFT EP/Z000157/1. The work of LA is partially supported by the National Natural Science Foundation of China, under Grant Number W2431012.

\end{document}